\documentclass[smallcondensed]{svjour3}
\usepackage[margin=1.5in]{geometry}
\usepackage{graphicx}
\usepackage{epstopdf}
\usepackage{amsmath,amssymb,amsfonts}
\usepackage{mathtools}
\usepackage{multirow}
\usepackage{multicol}
\usepackage{verbatim}
\usepackage{url}
\usepackage{comment}
\usepackage{mathtools}
\usepackage{epsf,pgf,graphicx}
\usepackage{color}
\usepackage{tikz,pgf}
\usepackage{diagbox}
\usepackage{makecell}
\usepackage{braket}
\usepackage{mathtools}
\usepackage{xcolor}
\usepackage{subfig}
\usepackage{hyperref}
\usepackage{ulem}
\usepackage{arydshln}
\usepackage[export]{adjustbox}

\usepackage{enumitem}
\newlist{abbrv}{itemize}{1}
\setlist[abbrv,1]{label=,labelwidth=1in,align=parleft,itemsep=0.1\baselineskip,leftmargin=!}
\usepackage{hyperref}
\hypersetup{colorlinks,
citecolor=blue,
urlcolor=blue,
linkcolor=blue
}
\usepackage{tabstackengine}
\stackMath
\usepackage{array}
\newcolumntype{M}[1]{>{\hbox to #1\bgroup\hss$}l<{$\egroup}}

\makeatletter
\newcommand\@brcolwidth{0.67em}

\def\@brarray[#1]{\array{r*\c@MaxMatrixCols {M{#1}}}}
\makeatother

\newtheorem{cons}{Construction}

\title{On Construction of Approximate Real Mutually Unbiased Bases for an infinite class of dimensions $d \not\equiv 0 \bmod 4$}

\makeatletter
\renewcommand*{\@fnsymbol}[1]{\ensuremath{\ifcase#1\or *\or \dagger\or \ddagger\or
   \mathsection\or \mathparagraph\or \|\or **\or \dagger\dagger
   \or \ddagger\ddagger \else\@ctrerr\fi}}
\makeatother

\author{Ajeet Kumar \and Rakesh Kumar \and Subhamoy Maitra \and Uddipto Mandal}
\institute{
A. Kumar \at
Applied Statistics Unit, Indian Statistical Institute, Kolkata, India,
\email{ajeetk52@gmail.com}
\and
R. Kumar \at
Applied Statistics Unit, Indian Statistical Institute, Kolkata, India,
\email{rkmath1729@gmail.com}
\and
S. Maitra \at
Applied Statistics Unit, Indian Statistical Institute, Kolkata, India,
\email{subho@isical.ac.in}
\and
U. Mandal \at
Indian Institute of Technology, Kharagpur, India,
\email{uddiptomandal2006.24@kgpian.iitkgp.ac.in}
}
\authorrunning{\textsc{Ajeet Kumar}, \textsc{Rakesh Kumar}, \textsc{and Subhamoy Maitra}}

\begin{document}
\maketitle
\begin{abstract}
It is known that real Mutually Unbiased Bases (MUBs) do not exist for any dimension $d > 2$ which is not divisible by 4. Thus, the next combinatorial question is how one can construct Approximate Real MUBs (ARMUBs) in this direction with encouraging parameters. In this paper, for the first time, we show that it is possible to construct $> \lceil \sqrt{d} \rceil$ many ARMUBs for certain odd dimensions $d$ of the form $d = (4n-t)s$, $t = 1, 2, 3$, where $n$ is a natural number and $s$ is an odd prime power. Our method exploits any available $4n \times 4n$ real Hadamard matrix $H_{4n}$ (conjectured to be true) and uses this to construct an orthogonal matrix ${Y}_{4n-t}$ of size $(4n - t) \times (4n - t)$, such that the absolute value of each entry varies a little from $\frac{1}{\sqrt{4n-t}}$. In our construction, the absolute value of the inner product between any pair of basis vectors from two different ARMUBs will be $\leq \frac{1}{\sqrt{d}}(1 + O(d^{-\frac{1}{4}})) < 2$, for proper choices of parameters, the class of dimensions $d$ being infinitely large. 
\end{abstract}
\section{Introduction}
\label{ch6:Introduction}
Mutually Unbiased Bases (MUBs) are important combinatorial objects that received serious attention in Quantum Information Processing. Such structures are useful in different aspects of Quantum Cryptology and Communications. For example, the list includes Quantum Key Distribution (QKD), Teleportation, Entanglement Swapping, Dense Coding, Quantum Tomography, etc. (see~\cite{MUB1} and the references therein). It is well known that given any finite dimensional Hilbert space $\mathbb{C}^d$, the number of MUBs can at most be $d+1$. In spite of considerable effort for several decades, the cases reaching the upper bound could only be constructed when $d$ is a prime power~\cite{IV81,wootters1989optimal,bandyopadhyay2002new,KL03,Sulc2007}. For a specific dimension $d$, other than the prime powers, constructing a larger number of MUBs (upper bounded by $d+1$) is believed to be one of the most challenging problems in Quantum Information Theory. The situation becomes more restrictive when we consider the elements of the MUBs to be real. In particular, there cannot be even two real MUBs for a dimension greater than 2 and not divisible by 4, as pair of MUBs can always be transformed into Hadamard matrix which is MUB with respect to the identity matrix but there doesn't exist real Hadamard matrix having dimensions greater than 2 and not divisible by 4.

Mathematical tools from various domains have already been exploited to construct the MUBs, for example the use of finite fields~\cite{wootters1989optimal,KL03} as well as maximal set of commuting bases~\cite{bandyopadhyay2002new}. For dimensions, that are not prime powers, constructing large number of MUBs remains elusive. At this point, various kinds of Approximate MUBs have been explored. There are results in this domain where such approximate structures have been constructed using character sums over Galois Rings or Galois Fields~\cite{AMUB2,AMUB1,AMUB-MixedCharacterSum,AMUB-CharacterSum,AMUB-GaloisRing,ARMUB-fromComplexAMUB}, combinatorial designs~\cite{ak21,ak22,ak23} and computational search methods~\cite{ak22a}.

When MUBs are constructed over $\mathbb{R}^d$, one can obtain real MUBs. This area of research connects with Quadratic Forms~\cite{cameron1991quadratic}, Association Schemes~\cite{ReMUB-AssociationSchemes,AssociationScheme-Coding}, Equi-angular Lines, Equi-angular Tight Frames over $\mathbb{R}^d$~\cite{ReMUB-FusionFrames}, Representation of Groups~\cite{ReMUB-GroupRepresentation}, Mutually Unbiased Real Hadamard Matrices, Bi-angular vectors over $\mathbb{R}^d$~\cite{Holzmann2010,Kharaghani2018,Best2015} and Codes~\cite{Calderbank1997}. As we have pointed out, large number of real MUBs are non-existent for most of the dimensions~\cite{MUB2}. In fact, only for $d = 4^s, s > 1$, one can obtain $\frac{d}{2} + 1$ many real MUBs which is maximum for any dimension $d$, whereas for most of the dimensions $d$, which are not perfect square, we have at best only $2$ real MUBs~\cite{MUB2}. One may refer to~\cite[Table 1]{MUB2} for a comprehensive view regarding this. This provides the motivation to construct Approximate Real MUBs (ARMUBs) as evident in recent literature~\cite{ak21,ak22,ARMUB-fromComplexAMUB}. Given $d = p_1^{n_1} p_2^{n_2} \ldots p_s^{n_s}$, the lower bound on the number of MUBs in $\mathbb{C}^d$ is $p_r^{n_r}  + 1$ where $p_r^{n_r} = \min\{{p_1^{n_1}, p_2^{n_2}, \ldots,p_s^{n_s}}\}$. Naturally, constructing a good number of MUBs for any composite dimension proved to be elusive even over $\mathbb{C}^d$. In fact, the number of such bases is very small~\cite{MUB2} when we consider the problem over the real vector space $\mathbb{R}^d$. 
This is summarized in~\cite[Theorem 3.2]{ReMUB-AssociationSchemes} too. In this regard, one must note that the there cannot be any pair of real MUBs for dimensions which are $> 2$ and not divisible by 4. This is immediately understood as there cannot be any real Hadamard matrix in such dimensions and existence of a pair of real MUBs provides a Hadamard matrix. With this motivation, in this paper, we consider construction of Approximate Real MUBs (ARMUBs) for dimensions $d \not\equiv 0 \bmod 4$ where $d > 2$. 

\subsection{Organization \& Contribution}
The organization of the paper is as follows. 
In Section~\ref{ch6:pre:ARMUB}, we define MUBs and the approximate variants. We outline the existing ideas with examples that will be used in our contributions. 
In Section~\ref{ch6:cont:ARMUB}, we construct the orthogonal matrices of orders $4n - t, \ t = 1, 2, 3$, from real Hadamard matrix ${H}_{4n}$ of order $4n$. We name these as $\epsilon$-Hadamard matrices. Finally, using such matrices and suitable RBDs, we construct ARMUBs for dimensions $d = q(4n-t)$, for $t = 1, 2, 3$. We will show that when $q$ is an odd prime power, and the order of $4n$ and $q$ are the same, then we can have order of $\sqrt{d}$ many ARMUBs with significantly low inner product values.

To compare with the earlier works, we will refer to~\cite{ak21,ak22,ak23}. 
\begin{itemize}
\item In~\cite{ak21}, the ARMUBs are constructed for the dimensions $d = (4x)^2$, where $x$ is a prime power. The number of such ARMUBs was $\frac{\sqrt{d}}{4}+1$, and the value of the inner products was $\leq \frac{4}{\sqrt{d}}$.
\item This has been improved in~\cite{ak22}, where ARMUBs for dimensions $d = q(q+1)$, when prime power $q \equiv 3 \bmod 4$ were considered. The inner products were improved to $< \frac{2}{\sqrt{d}}$, and $\lceil\sqrt{d}\rceil$ many such ARMUBs could be obtained. In certain cases results were available for 
$d = sq^2$, where $q$ is a prime power and $sq \equiv 0 \bmod 4$.
\item Later APMUBs (Almost Perfect) MUBs were studied in~\cite{ak23}, and several constructions over complex numbers were considered. The constructions could be achieved for reals, but in those cases the dimensions are again divisible by 4, as they require existence of real Hadamard matrices. 
\end{itemize}
In this paper we modify the real Hadamard matrices to have real unitary ones of dimensions $d = q(4n-t)$, for $t = 1, 2, 3$, which we refer as $\epsilon$-Hadamard. These are used in conjunction with RBDs in line of the constructions in~\cite{ak22,ak23} to obtain new constructions of $\beta$-ARMUBs for the dimensions $d > 2$ and $d \not\equiv 0 \bmod 4$, with $\beta < 2$. In algebraic terms, this bound is 
$\beta \leq \frac{1}{\sqrt{d}}(1 + O(d^{-\frac{1}{4}}))$.
When $s$ is an odd prime power, and the order of $4n-t$ and $s$ are same, we obtain $O(\sqrt{d})$ many ARMUBs for such dimensions. Note that even a pair of real MUBs are not available for the dimensions $d > 2$ and $d \not\equiv 0 \bmod 4$ and thus the approximate versions are the only relevant solutions in this direction, that we first present in this paper.

\section{Preliminaries \& Background}
\label{ch6:pre:ARMUB}
Let us first formally introduce the MUBs.
\begin{definition}
\label{ch6:adef1}
Let $\mathbb{C}^{d}$ be a $d$-dimensional complex vector space. Now consider two orthonormal bases, 
$M_l = \left\{\ket{\psi_1^l}, \ket{\psi_2^l}, \ldots, \ket{\psi_{d}^l}\right\} 
\mbox{ and } M_m = \left\{\ket{\psi_1^m}, \ket{\psi_2^m}, \ldots, \ket{\psi_{d}^m}\right\}$, 
in $\mathbb{C}^{d}$. These two bases will be called {\sf Mutually Unbiased} if we have $\left|\braket{\psi_i^l | \psi_j^m}\right| = \displaystyle\frac{1}{\sqrt{d}}, \ \forall i, j \in \left \{1, 2, \ldots, d\right\}$.   
The set $\mathbb{M} = \left\{M_1, M_2, \ldots, M_r\right\}$ consisting of such orthonormal bases will form MUBs of size $r$, 
if every pair in the set is mutually unbiased.
\end{definition}
When the two bases are not mutually unbiased, but almost close to that, then we call them {\sf Approximate MUBs} (AMUBs)~\cite{ak21,ak22}.
\begin{definition}
\label{ch6:adef2}
The two bases will be called {\sf Approximately Mutually Unbiased} if we have
\begin{equation}
    \left|\braket{\psi_i^l | \psi_j^m}\right| \approx \displaystyle\frac{1}{\sqrt{d}}, \ \forall i, j \in \left \{1, 2, \ldots, d\right\} .   
\end{equation}
The set $\mathbb{M} = \left\{M_1, M_2, \ldots, M_r\right\}$ consisting of such orthonormal bases will form {\sf Approximate MUBs} (AMUBs) of size $r$, 
if every pair in the set is approximately mutually unbiased. This approximation needs to be formalized and it is expected that the deviation from $\frac{1}{\sqrt{d}}$ should be as small as possible.
When the bases are real, we call them {\sf Approximate Real MUBs (ARMUBs)}. 
\end{definition}
In this regard, there are two more definitions.
\begin{definition}
    A set of $r$ orthonormal bases $M_{i}$, $0\leq i \leq r-1$ of $\mathbb{R}^d$ are called $\beta$-ARMUBs if for any two bases $M_l$ and $M_m$, we have 
\begin{equation}
    \left|\braket{\psi_i^l | \psi_j^m}\right| \leq \displaystyle\frac{\beta}{\sqrt{d}}, \ \forall i, j \in \left \{1, 2, \ldots, d\right\} .   
\end{equation}
\end{definition}
\begin{definition} ~\cite{ak23}
\label{ch6:def:APMUB}
The set $\mathbb{M} = \{M_1, M_2, \ldots, M_r\}$ will be called {\sf Almost Perfect MUBs} (APMUBs) if $\Delta  = \left \{0,\frac{\beta}{\sqrt{d}}  \right \}$, i.e., the set contains just two values, such that  $\beta = 1+O(d^{-\lambda}) \leq 2$, for any real $\lambda > 0$. When the bases are real, we call them {\sf Almost Perfect Real MUBs} (APRMUBs). 
\end{definition}
While studying the approximate MUBs over $\mathbb{C}^d$, a few constructions were proposed in~\cite{AMUB2,AMUB1,AMUB-CharacterSum}. However, such techniques cannot be applied in construction of real MUBs. In this direction, a construction has been proposed in~\cite{ak21} for Approximate Real MUBs using real Hadamard matrices. It has been described in~\cite{ak21} that $\frac{\sqrt{d}}{4}+1$ ARMUBs with maximum value of the inner product as $\frac{4}{\sqrt{d}}$ could be achieved for $d = (4q)^2$, where $q$ is a prime. In~\cite{ak22}, it could be shown that such results can be generalized as well as improved exploiting the Resolvable Block Designs (RBDs). The earlier result of~\cite{ak21} could be generalized in~\cite{ak22} for $d = sq^2$, where $q$ is a prime power. Further, the parameters could be improved too in~\cite{ak22}. It has been shown in~\cite{ak22} that for $d = q(q+1)$, where $q$ is a prime power and $q \equiv 3 \bmod 4$, it is possible to construct $\lceil{\sqrt{d}}\rceil = q+1$ many ARMUBs with maximum value of the inner product upper bounded by $\frac{2}{\sqrt{d}}$, between vectors belonging to different bases. Therefore, the improvement in the result of~\cite{ak22} is two-fold. First, the number of MUBs is greater and second, the maximum of the inner product values is lower compared to~\cite{ak21}. Works in a similar direction with advancement considering APRMUBs have been recently published in~\cite{ak23}, but that work did not consider odd dimensions too.

One may note that each MUB in the space $C^{d}$ consists of $d$ orthogonal unit vectors which, collectively, can be thought of as a unitary $d \times d$ matrix. Two (or more) MUBs thus correspond to two (or more) unitary matrices, one of which can always be mapped to the identity $I$ of the space $C^{d}$, using a unitary transformation. For example, suppose we have $r$ many MUBs $\{M_{1}, M_{2}, M_{3}, \dots, M_{r}\}$ in $C^{d}$ where $r \leq d+1$ and also we can consider them as a $r$ many $d \times d$ unitary matrices. If we multiply $M^{-1}_1$ to each of the matrices from left then one can obtain $\{I,M_{1}^{-1}M_2,M_{1}^{-1}M_3,\dots, M_{1}^{-1}M_{r}\} $ as the transformed set of MUBs. As the inverse of any unitary matrix is equal to its conjugate transpose, to obtain $ M_{i}^{-1}M_j$, for $i \neq j$, we are considering inner products of each column of the two matrices in the set of MUBs. Thus, the modulus of each element of the product matrix will be $\frac{1}{\sqrt{d}}$. Taking $\frac{1}{\sqrt{d}}$ common, the modulus of each of the elements will be 1, i.e., we will have complex Hadamard matrices.

\subsection{Some basic examples related to MUBs}
We have already explained MUBs in Definition~\ref{ch6:adef1}. Now we will proceed with our explanations with small dimensions. For $d=2$, where we have 3 MUBs, 
$M^{(2)}_0 = \left[\begin{array}{cc}
1 & 0\\
0 & 1
\end{array}\right]$, 
$M^{(2)}_1 = 
\left[\begin{array}{cc}
\frac{1}{\sqrt{2}} & \frac{1}{\sqrt{2}}\\
\frac{1}{\sqrt{2}} & -\frac{1}{\sqrt{2}}
\end{array}\right] = 
\frac{1}{\sqrt{2}}
\left[\begin{array}{cc}
1 & 1\\
1 & -1
\end{array}\right]$ and
$M^{(2)}_2 = 
\left[\begin{array}{cc}
\frac{1}{\sqrt{2}} & \frac{i}{\sqrt{2}}\\
\frac{1}{\sqrt{2}} & -\frac{i}{\sqrt{2}}
\end{array}\right] = 
\frac{1}{\sqrt{2}}
\left[\begin{array}{cc}
1 & i\\
1 & -i
\end{array}\right]$, the first two being real. However, this is not possible for the dimension $d=3$. There are indeed four MUBs, that can be represented as $M^{(3)}_0 = \left[\begin{array}{ccc}
1 & 0 & 0\\
0 & 1 & 0\\
0 & 0 & 1\\
\end{array}\right]$,
$M^{(3)}_1 = \frac{1}{\sqrt{3}} \left[\begin{array}{ccc}
1 & 1 & 1\\
\omega^2 & 1 & \omega\\
\omega^2 & \omega & 1\\
\end{array}\right]$,

$M^{(3)}_2 = \frac{1}{\sqrt{3}} \left[\begin{array}{ccc}
1 & 1 & 1\\
1 & \omega & \omega^2\\
1 & \omega^2 & \omega\\
\end{array}\right]$,
$M^{(3)}_3 = \frac{1}{\sqrt{3}} \left[\begin{array}{ccc}
1 & 1 & 1\\
\omega & \omega^2 & 1\\
\omega & 1 & \omega^2\\
\end{array}\right]$, but except the identity, none of them are reals. Thus the a pair of real MUBs is not available in this scenario.

Let us now present a set of five MUBs for $d=4$:

$M^{(4)}_0 = \left[\begin{array}{cccc}
1 & 0 & 0 & 0\\
0 & 1 & 0 & 0\\
0 & 0 & 1 & 0\\
0 & 0 & 0 & 1\\
\end{array}\right]$,
$M^{(4)}_1 = \frac{1}{2} \left[\begin{array}{cccc}
1 & 1 & 1 & 1\\
1 & -1 & 1 & -1\\
1 & 1 & -1 & -1\\
1 & -1 & -1 & 1\\
\end{array}\right]$,
$M^{(4)}_2 = \frac{1}{2} \left[\begin{array}{cccc}
1 & -1 & -i & -i\\
1 & -1 & i & i\\
1 & 1 & i & -i\\
1 & 1 & -i & i\\
\end{array}\right]$,

$M^{(4)}_3 = \frac{1}{2} \left[\begin{array}{cccc}
1 & -i & -i & -1\\
1 & -i & i & 1\\
1 & i & i & -1\\
1 & i & -i & 1\\
\end{array}\right]$,
$M^{(4)}_4 = \frac{1}{2} \left[\begin{array}{cccc}
1 & -i & -1 & -i\\
1 & -i & 1 & i\\
1 & i & -1 & i\\
1 & i & 1 & -i\\
\end{array}\right]$.

There is only a pair of real MUBs in this case. Now the question is can we extend this, and the answer is affirmative with the following example with 3 real MUBs. 

$M'_{1} =
        \frac{1}{\sqrt{2}}\begin{bmatrix} 
	1 & 1 & 0 & 0 \\
	1 & -1 & 0 & 0\\
	0 & 0 & 1 & 1\\
        0 & 0 & 1 & -1\\
        \end{bmatrix},
M'_{2}=\frac{1}{\sqrt{2}}\begin{bmatrix} 
	1 & 1 & 0 & 0 \\
	0 & 0 & 1 & 1\\
	1 & -1 & 0 & 0\\
        0 & 0 & 1 & -1\\
	\end{bmatrix},
 M'_{3}=\frac{1}{\sqrt{2}}\begin{bmatrix} 
	1 & 1 & 0 & 0 \\
	0 & 0 & 1 & 1\\
	0 & 0 & 1 & -1\\
        1 & -1 & 0 & 0\\
	\end{bmatrix}$.
	
Consider the multiplication by the inverse of $M'_{1}$ from the left hand side. Here $M'_1 = (M'_1)^{\dagger} = (M'_1)^{-1}$. Then we obtain, $(M''_1)^{-1} = (M'_{1})^{-1} M'_{1} = I$:
 
 $M''_{2} = (M'_{1})^{-1} M'_{2} = (M'_{1})^{\dagger} M'_{2} = \frac{1}{2} \begin{bmatrix} 
	1 & 1 & 0 & 0 \\
	1 & -1 & 0 & 0\\
	0 & 0 & 1 & 1\\
        0 & 0 & 1 & -1\\
        \end{bmatrix} 
       \begin{bmatrix} 
	1 & 1 & 0 & 0 \\
	0 & 0 & 1 & 1\\
	1 & -1 & 0 & 0\\
        0 & 0 & 1 & -1\\
	\end{bmatrix}
       = \frac{1}{2}\begin{bmatrix} 
	1 & 1 & 1 & 1 \\
	1 & 1 & -1 & -1\\
	1 & -1 & 1 & -1\\
        1 & -1 & -1 & 1\\
	\end{bmatrix}$ and
       
        $M''_{3}= (M'_{1})^{-1}M'_{3}= (M'_{1})^{\dagger}M'_{3}$= $\frac{1}{2} \begin{bmatrix} 
	1 & 1 & 0 & 0 \\
	1 & -1 & 0 & 0\\
	0 & 0 & 1 & 1\\
        0 & 0 & 1 & -1\\
        \end{bmatrix}
         \begin{bmatrix} 
	1 & 1 & 0 & 0 \\
	0 & 0 & 1 & 1\\
	0 & 0 & 1 & -1\\
        1 & -1 & 0 & 0\\
	\end{bmatrix}
       = \frac{1}{2}\begin{bmatrix} 
	1 & 1 & 1 & 1 \\
	1 & 1 & -1 & -1\\
	1 & -1 & 1 & -1\\
        -1 & 1 & 1 & -1\\
	\end{bmatrix}$
	as Hadamard matrices.
As the examples contain real values only, the notation $\dagger$ works as simple transpose here. The structures of $M'_{1}, M'_{2}, M'_{3}$ can be achieved by the following method that has been used extensively in~\cite{ak22,ak23} in design of Approximate MUBs.

\subsection{Constructing MUBs through RBD}
\label{ch6:rbdsec}
Now we briefly explain combinatorial block design as we use that structure extensively in this paper (for deeper details refer to~\cite{stinson2007combinatorial}).	
\begin{definition}
A combinatorial block design is a pair $(X, A)$, where $X$ is a set of elements, called elements, and $A$ is a collection of non-empty subsets of $X$, called blocks. A combinatorial design is called simple, if there is no repeated block in $A$. Further, the design $(X, A)$ is called a Resolvable Block Design (RBD), if $A$ can be partitioned into $r \geq 1$ parallel classes, called resolutions, where each parallel class in the design $(X, A)$ is a subset of the disjoint blocks in $A$, whose union is $X$.
\end{definition}
For RBD$(X, A)$, with $|X| = d$, we will indicate the elements of $X$ by simple numbering, i.e., $X = \{1, 2, 3, \ldots, d\}$. Here $r$ will denote the number of parallel classes in RBD and the parallel classes will be represented by $P_1, P_2, \ldots, P_r$. The blocks in the $l$-th parallel class will be represented by $\{b_1^l, b_2^l, \ldots, b_s^l \}$, indicating that the $l$-th parallel class has $s$ many blocks. The notation $b_{ij}^l $ would represent the $j^{th}$ element of the $i^{th}$ block in the $l^{th}$ parallel class.

In our analysis we will be using the RBDs with constant block size and let us denote the block size by $k$. Further, we denote the number of blocks in a parallel class of RBD by $s$. Since in our analysis we are making of use of RBDs with constant block size, hence each parallel class will always have $s$ many blocks and $|X| = d = k \cdot s$. 

For example, consider $X = \{1, 2, 3, 4\}$, $A = \{ \{1, 2\}, \{3, 4\}, \{1, 3\}, \{2, 4\}, \{1, 4\}, \{2, 3\}\}$. Thus, $d = |X| = 4$, each block of $A$ is a set of size $k = 2$, and any element of $X$ appears in exactly $\lambda = 3$ blocks.  There are 3 parallel classes, $P_1 = \{ \{1, 2\}, \{3, 4\} \}, P_2 = \{ \{1, 3\}, \{2, 4\} \}$, and $P_3 = \{ \{1, 4\}, \{2, 3\} \}$ and here $s=2$.
 
In every block, we will arrange the elements in increasing order, and we will follow this convention throughout the paper, unless mentioned specifically. Thus $b_{ij}^l \leq b_{i,j+1}^l, \ \forall j$. This will be important to revisit when we convert the parallel classes into orthonormal bases. Another important parameter for our construction is the value of  the maximum number of common elements between any pair of blocks from different parallel classes. We denote this positive integer by $\mu$. Note that $\mu \geq 1$ for any RBD, with $r \geq 2$. 

When the conditions among two different bases are relaxed such that $\left|\braket{\psi_i^l | \psi_j^m}\right|$ can take values other than $\frac{1}{\sqrt{d}}$, then we consider the approximate version of this problem. This is due to the fact that, for most of the dimensions which are not power of some primes, obtaining a large number of MUBs reaching the upper bound is elusive. In this context we denote $\Delta$ for the set containing different values of $\left|\braket{\psi_i^l | \psi_j^m}\right|$ for $l \neq m$. Let us now explain the construction presented in~\cite{ak22,ak23}.
\begin{cons}
\label{ch6:cons:1} \
\begin{enumerate}
\item In a design $(X, A)$, choose the elements of  $X$ as some orthonormal basis vectors of $\mathbb{C}^d$. We will interpret these orthonormal vectors as columns. 
That is, when $|X|= d$ then one can write
$X = \{\ket{\psi_1}, \ket{\psi_2}, \ldots, \ket{\psi_{d}}\}$, such that $\braket{\psi_i | \psi_j}= \delta_{ij}$. $A$ contains certain subsets of $X$, called blocks.

\item Let $P = \{b_1, b_2, \ldots, b_s\}$ be one of the parallel classes of the design $(X, A)$, where the blocks are disjoint containing elements 
from $X$. Let there be $r$ parallel classes and their union will be $A$.

\item Consider one of the blocks $b_t = \{\ket{\psi_{t_1}} ,\ket{\psi_{t_2}}, \ldots, \ket{\psi_{t_{n_t}}}\}  \in P$ and let $|b_t| = n_t$. 
Corresponding to this block, choose any $n_t \times n_t$ unitary matrix $u^t_{i,j}$. This unitary matrix should be realized by Hadamard or Hadamard like ones.

\item Next construct $n_t$ many vectors in the following manner, using $b_t$ and $u^t_{i,j}$.
\begin{equation*}
\ket{\phi^t_i} =  u^t_{i,1}\ket{\psi_{t_1}} + u^t_{i,2} \ket{\psi_{t_2}} + \ldots + u^t_{i,n_t} \ket{\psi_{t_{n_t}}} = 
\sum_{k=1}^{n_t} u^t_{i,k} \ket{\psi_{t_k}}: i = 1, 2, \ldots, n_t.
\end{equation*}

\item In a similar fashion, corresponding to each block $b_j \in B$, construct $n_j$ many vectors where $|b_j| = n_j$, 
using any $n_j \times n_j$ unitary matrix. Since $\sum_{j=1}^s n_j = d$, we will get exactly $d$ many vectors. 
\end{enumerate}
\end{cons}

To explicitly demonstrate the construction of real MUBs in $\mathbb{R}^4$, 
using Resolvable $(4,2,1)$-BIBD, let $X =\left\{1,2,3,4\right\}$ be the four standard basis 
vectors in $\mathbb{R}^4$ and let $A = \left\{ P_1, P_2, P_3\right\}$ be the three
parallel classes of the resolvable design. Explicitly, one such design would be:
\begin{align*}
P_1 = \left\{ (1,2),(3,4) \right\} \,\,\,\,\, P_2 =\left \{ (1,3),(2,4) \right\} \,\,\,\,\, P_3 = \left \{ (1,4),(2,3) \right\}.
\end{align*}
Now using $H_2 = \frac{1}{\sqrt{2}}\begin{pmatrix}
     1 &   1 \\
    1  &  -1
\end{pmatrix}$ for each block of parallel class, and exploiting Construction \ref{ch6:cons:1}, 
we obtain three set of orthonormal basis vectors corresponding to each parallel class as follows.
\begin{equation*}
M_1 = \frac{1}{\sqrt{2}} \begin{pmatrix}
      1&1&0&0 \\
      1&-1 &0 &0   \\
      0&0  &1&1 \\
      0 &0  &1 &-1  \\
\end{pmatrix},
M_2 = \frac{1}{\sqrt{2}} \begin{pmatrix}
           1&1 &0 &0  \\
           0 &0 &1 &1  \\
           1&-1 &0 &0\\
            0 &0 &1 &-1\\
\end{pmatrix},
M_3 = \frac{1}{\sqrt{2}} \begin{pmatrix}
           1&1 &0 &0  \\
            0&0 &1 &1  \\
            0&0 &1 &-1\\
            1 &-1 &0 &0\\
\end{pmatrix}.
\end{equation*}
The columns of $\{M_1, M_2, M_3\}$ form the orthonormal basis vectors, and these 
orthonormal bases are MUBs in $\mathbb{R}^4$. Note that maximum number of real MUBs in 
$d=4^s , s\in \mathbb{N}$ is equal to $\frac{d}{2} + 1$ \cite{MUB2}. Hence for $d = 4$, 
the three MUBs constructed are also maximal for $\mathbb{R}^4$.
However, such a situation will not happen for any integer. In fact, as we have already discussed, such real Hadamard matrices are available for $d = 2$ and when $d$ is divisible by 4 (conjectured). That is why, we try to have Approximate Real MUBs or ARMUBs. We present this in the following contributory section.

\section{Our Construction}
\label{ch6:cont:ARMUB}
In this section we present our technique to have ARMUBs, which were not presented earlier in literature. For construction of ARMUBs, similar to Construction~\ref{ch6:cons:1}~\cite{ak22,ak23}, we make use of RBD$(X, A)$, with $|X|= d = k\times s$, where $s > k$ and $s, k$ are as close as possible, so that both $s, k$ are of $O(\sqrt{d})$. Since $d > 2$ is not multiple of $4$, thus $k$ or $s$ cannot be a multiple of $4$. For constructing ARMUBs corresponding to each parallel class of an RBD, we exploit suitable orthogonal matrix $Y$, which we call (as in Definition~\ref{ch6:defH} later) $\epsilon$-Hadamard Matrix of order $k$, modifying a real Hadamard one of order $k+t$ (where $k+t$ is divisible by 4) such that $|(Y)_{ij}|$ is very close to
$\frac{1}{\sqrt{k}}$. The Hadamard conjecture tells that in such cases we have the existence, and thus we can consider $t= \{1, 2, 3\}$ for our purpose. Let us now consider how do we modify the Hadamard matrices.

\subsection{Modifying real Hadamard matrices}
As discussed, we will construct orthogonal matrix of order $k \in \{4n-1, 4n-2, 4n- 3\}$ from Hadamard matrix of order $4n$, using the method presented below. We will call such orthogonal matrices as $\epsilon$-Hadamard Matrices. We define the notion of $\epsilon$-Hadamard matrices of order $k$, which in some sense can be considered as close counterpart of Hadamard in such order, where no Hadamard matrix exists ($k \neq 4n$) and will be helpful in constructing $\beta$-ARMUBs with good properties. The basic property of Hadamard matrix which is helpful in such construction, is the same absolute value of all the entires, which is $\frac{1}{\sqrt{4n}}$. Taking cue from this, we intend to construct orthogonal matrix for order $k \neq 4n$, such that the absolute value of the entries of the matrix are very close to each other. In this regard, let us present the following definition for our purpose.

\begin{definition}
\label{ch6:defH}
For an orthogonal matrix $\mathbb{O}_{k\times k}$, if the absolute value of each entry of this matrix $|(\mathbb{O}_{k\times k})_{ij}| = \frac{1 + \epsilon_{ij}}{\sqrt{k}}$, where $\epsilon_{ij} = \pm O(k^{-\lambda}) < 1$, such that $\lambda > 0$, then we call the orthogonal matrix $\mathbb{O}_{k\times k}$ an $\epsilon$-Hadamard matrix, when $\epsilon = \max_{i, j} |\epsilon_{ij}|$.
\end{definition}
If there are choices of many such matrices, then we will consider the one with least $\epsilon$ available at hand. When $k$ is a multiple of 4, assuming Hadamard conjecture, it is obvious that the minimal value of $\epsilon$ is zero, but when $k > 2$ is not a multiple of 4, finding the global minima of $\epsilon$ appears to be a very challenging problem. However, we first show in Theorem~\ref{ch6:th:ReducH} that such result with $\epsilon < 1$ can be achieved. Later, in Section~\ref{ch6:subsection:cases_for_hadamard}, we will consider various $U$ and correspondingly, we will obtain the orthogonal matrix with the least $\epsilon$. 

We now have the following technical result. Please note that with certain abuse of notations, we write $U_{t \times t}$ to express that it is a $t \times t$ matrix, but while explaining we only write $U$ as we can refer the $(i, j)$-th element as subscript such as $U_{ij}$ or $(U)_{ij}$. 
\begin{lemma}
\label{ch6:Lem:Reduc}
Let $N= 
\left[ 
\begin{array}{cc}
U_{t \times t} & V_{t \times (m-t)} \\
W_{(m-t)\times t} & D_{(m-t) \times (m-t)} 
\end{array} 
\right]_{m \times m}$
be an $m \times m$ orthogonal matrix, containing block matrices $U$, $V$, $W$, and $D$ of sizes as indicated.
If either $(\mathbb{I}+U)$ or $(\mathbb{I}-U)$ is a nonsingular matrix, then 
$Y_1 = D - W(\mathbb{I}+U)^{-1}V \quad \text{and} \quad Y_2 = D + W( \mathbb{I}-U)^{-1}V$
are both orthogonal matrices of order $m-t$.
\end{lemma}
\begin{proof}
For a matrix $A$, the transpose is referred as $A^T$.
Since $N N^T = \mathbb{I}_{m \times m}$, we have
\begin{align*}
& UU^T + VV^T = \mathbb{I}_{t \times t}, \\
& UW^T + VD^T = \mathbb{O}_{t \times (m-t)}, \\
& WU^T + DV^T = \mathbb{O}_{(m-t) \times t}, \\
& WW^T + DD^T = \mathbb{I}_{(m-t) \times (m-t)}.
\end{align*}
Suppose $(\mathbb{I}+U)^{-1}$ exists. Consider, the matrix 
$Y_1 = D - W( \mathbb{I}+U)^{-1}V$.
Then,
\begin{align*}
Y_1Y_1^T &= \left(D - W(\mathbb{I}+U)^{-1}V\right)\left(D - W( \mathbb{I}+U)^{-1}V\right)^T \\
&= \left(D - W(\mathbb{I}+U)^{-1}V\right)\left(D^T - V^T(\mathbb{I}+U^T)^{-1}W^T\right) \\
&= DD^T - DV^T (\mathbb{I}+U^T)^{-1}W^T - W(\mathbb{I}+U)^{-1}VD^T \\
&\quad + W(\mathbb{I}+U)^{-1}VV^T(\mathbb{I}+U^T)^{-1}W^T \\
&= \mathbb{I} - WW^T + WU^T( \mathbb{I}+U^T)^{-1}W^T + W(\mathbb{I}+U)^{-1}UW^T \\
&\quad + W( \mathbb{I}+U)^{-1}VV^T(\mathbb{I}+U^T)^{-1}W^T \\
&= \mathbb{I} - W\left[\mathbb{I} - U^T( \mathbb{I} + U^T)^{-1} - (\mathbb{I}+ U)^{-1}U - ( \mathbb{I}+U)^{-1}VV^T(\mathbb{I}+U^T)^{-1}\right]W^T.
\end{align*}
Now consider the expression inside the brackets:
\begin{align*}
&\mathbb{I} - U^T( \mathbb{I}+ U^T )^{-1} - (\mathbb{I}+U)^{-1}U - (\mathbb{I}+U)^{-1}VV^T( \mathbb{I}+ U^T)^{-1} \\
&= ( \mathbb{I}+U)^{-1}( \mathbb{I}+U)(\mathbb{I} + U^T)( \mathbb{I} + U^T)^{-1}
 - ( \mathbb{I}+U)^{-1}( \mathbb{I}+U)U^T( \mathbb{I}+U^T)^{-1} \\ 
&\quad - (\mathbb{I} +U)^{-1}U( \mathbb{I}+U^T)( \mathbb{I}+U^T)^{-1}
- ( \mathbb{I}+U)^{-1}VV^T(\mathbb{I}+U^T)^{-1} \\
&= (\mathbb{I}+U)^{-1} \left[ ( \mathbb{I}+U)( \mathbb{I}+U^T) - ( \mathbb{I}+U)U^T - U( \mathbb{I}+U^T) - (\mathbb{I} - UU^T) \right] (\mathbb{I}+U^T)^{-1}.
\end{align*}
Since,
\begin{align*}
&(\mathbb{I}+U)( \mathbb{I}+U^T) - (\mathbb{I}+U)U^T - U(\mathbb{I}+U^T) - (\mathbb{I} - UU^T) \\
&=  \mathbb{I} + U^T + U  + UU^T - U^T- UU^T  -  U - UU^T -\mathbb{I} + UU^T = \mathbb{O},
\end{align*}
we obtain,
$Y_1Y_1^T = \mathbb{I}$.
Similarly in above line, one can also show that if $( \mathbb{I}-U)$ is invertible, then $Y_2 = D + W( \mathbb{I}-U)^{-1}V $ is an Orthogonal matrix of size $(m-t)\times (m-t)$. \qed
\end{proof}
Note that the inverse of $(\mathbb{I}+ U)$ will exist if eigen-value of $U$ is not equal to $-1$ and similarly inverse of $(\mathbb{I} - U)$ will exist if eigen value of $U$ is not equal to $1$. In fact, $U$ and $(\mathbb{I}\pm U)^{-1}$ commutes. Hence if $\nu$ is the eigen value of $U$ then $\frac{1}{\nu\pm 1}$ is the eigen vale of $(\mathbb{I}\pm U)^{-1}$. Hence, assuming the inverse of $(\mathbb{I}\pm U)$  exists, there is a simple form for the inverse of the matrix $(\mathbb{I} \pm U)$ in terms of power $U$, when $||U|| < 1$. This states that $(\mathbb{I} \pm U)^{-1} = \mathbb{I} \mp U + U^2 \mp U^3 + U^4 \ldots$ 

A particular form of matrix $(\mathbb{I} \pm U)$ is relevant  for demonstrating specific constructions, when it is of the form $(\mathbb{I} + \frac{1}{\alpha}X)$, where $X$ is a $t \times t$ matrix. Towards this we have following result.
\begin{lemma}  
\label{ch6:lemma:inverse}
Let $\kappa, \gamma, \alpha, \vartheta$ be real.
\begin{enumerate}
\item If $X^2 = \kappa \mathbb{I} + \gamma X$, then the  inverse of the matrix $\left(\mathbb{I} + \frac{1}{\alpha} X\right)$ exists if $\alpha^2 + \gamma \alpha - \kappa \neq 0$ and is given by  $ \frac{\alpha( \alpha +\gamma)}{\alpha^2 + \gamma \alpha - \kappa}\left(\mathbb{I} - \frac{1}{\alpha + \gamma} X \right)$.

\item If $X^3 = \kappa \mathbb{I} + \gamma X + \vartheta X^2 $, then the inverse of the matrix $\left(\mathbb{I} + \frac{1}{\alpha} X \right)$ exists if $\alpha^3 + \vartheta \alpha^2 - \gamma \alpha + \kappa \neq 0$ and is given by  $ \frac{\alpha(\alpha( \alpha +\vartheta) - \gamma)}{\alpha^3+ \vartheta \alpha^2 - \gamma \alpha + \kappa}\left(  \mathbb{I} -  \frac{ \alpha +\vartheta}{\alpha( \alpha +\vartheta) - \gamma } X + \frac{1}{\alpha( \alpha +\vartheta) - \gamma} X^2\right)$.
\end{enumerate}
\end{lemma}
\begin{proof} The proofs are as follows.
\begin{enumerate}
\item When $X^2 = \kappa \mathbb{I} + \gamma X$ consider $\left( \alpha \mathbb{I} +  X\right)  \left( (\gamma+ \alpha)\mathbb{I}- X \right) 
= \alpha (\alpha + \gamma)  \mathbb{I} + \gamma X - X^2    = (\alpha^2 + \gamma \alpha - \kappa)\mathbb{I}$. Hence if $\alpha^2 + \gamma \alpha - \kappa \neq 0$, the first result follows. 

\item When $X^3 = \kappa \mathbb{I} + \gamma X + \vartheta X^2 $, consider $\left( \alpha \mathbb{I} +  X\right) \left( (\alpha( \alpha +\vartheta) - \gamma) \mathbb{I} - ( \alpha +\vartheta) X + X^2\right)
=  (\alpha^3 + \vartheta \alpha^2- \gamma \alpha)  \mathbb{I} -  \gamma X  - \vartheta X^2  + X^3  
= (\alpha^3+ \vartheta \alpha^2 - \gamma \alpha + \kappa) \mathbb{I}$ . Hence if $\alpha^3+ \vartheta \alpha^2 - \gamma \alpha + \kappa \neq 0$, the second result follows. \qed
\end{enumerate}
\end{proof}
In the above lemma, if $X$ satisfies the item 1, then we should use that, and one may note that the solution from item 2 will be the same. In case, item 1 is not satisfied, then we should exploit item 2. This is because, the computation in the second case will be more involved.

Note that when $X$ is $2\times 2$ matrix, we will have $X^2 = \gamma X + \kappa \mathbb{I}$, where $\gamma = \text{Tr}(X)$ and $\kappa = -\text{Det}(X)$. When $X$ is a $3 \times 3$ matrix, $X^3 = \vartheta X^2 + \gamma X + \kappa \mathbb{I} $ where $\vartheta = \text{Tr}(X)$, $ \gamma = \sum_{i>j}  (X)_{ij} (X)_{ji} - (X)_{ii}(X)_{jj}$ and $\kappa = \text{Det}(X)$.

Let us now concentrate on eigen value characterization.
When $X^2 = \kappa \mathbb{I} + \gamma X$ then the eigen values of $X$ are the roots of the equation $\lambda^2 -\gamma \lambda - \kappa =0$, and hence has maximum two different values, say $\{\lambda_1, \lambda_2\}$. Thus, $X$ is similar to  diagonal matrix $\text{Diag}(\lambda_1, \ldots, \lambda_1 , \lambda_2, \ldots \lambda_2)$. Similarly when $X^3 = \kappa \mathbb{I} + \gamma X + \vartheta X^2$ then eigen values of $X$ are the roots of the equation $\lambda^3 - \vartheta \lambda^2 -\gamma \lambda - \kappa = 0$ and in this case  $X$ has maximum three different e-values, say $\{\lambda_1, \lambda_2, \lambda_3 \}$ and hence $X$ is similar to  diagonal matrix $\text{Diag}(\lambda_1, \ldots \lambda_1 , \lambda_2, \ldots \lambda_2, \lambda_3, \ldots \lambda_3)$.

We use this and apply the above Lemma \ref{ch6:Lem:Reduc} for the case when $N$ is a real Hadamard matrix of order $m = 4n$, denoting it by $H_{4n}$ and obtain the following result.

Note that, in Lemma \ref{ch6:Lem:Reduc}, there is no $\frac{1}{\sqrt{4n}}$ before the matrix, that we start using from the following result. This is by abuse of notation. 
Lemma \ref{ch6:Lem:Reduc} is only used to prove the orthogonality and thus, we did not use the constant term outside as in the Hadamard matrix. Thus, the $U, V, W, D$ as follows actually differ from those of Lemma \ref{ch6:Lem:Reduc} by a constant multiplier.

\begin{theorem}
\label{ch6:th:ReducH}
Let $H_{4n}= \frac{1}{\sqrt{4n}} \left[ \begin{tabular}{ c  c }
$U_{t\times t} $& ${V}_{t \times (4n-t)}$ \\
$W_{(4n-t)\times t}$ & $D_{(4n-t) \times (4n-t)}$ \\
\end{tabular} \right]_{n \times n}$. 
If $t < \sqrt{n}$ then $Y_1= D - W(\mathbb{I} + U)^{-1}$ and $Y_2= D + W(\mathbb{I} - U)^{-1}V$  are  $\epsilon$-Hadamard matrices of order $(4n-t)$, such that  $\frac{1}{\sqrt{4n}} \left(1- \frac{t}{\sqrt{4n} -t} \right)  \leq  |(Y_{1,2})_{ij}| \leq \frac{1}{\sqrt{4n}} \left(1+ \frac{t}{\sqrt{4n} -t} \right)$ with $\epsilon =  \frac{t}{\sqrt{k} } + \mathcal{O}(k^{-1}) \leq 1$. 
\end{theorem}
\begin{proof}
Here $U,V,W$ and $D$ sub sub matrix of a Hadamard matrix $H_{4n}$ as explained in Lemma~\ref{ch6:Lem:Reduc}. We have $(U)_{ij} = \frac{\pm 1}{\sqrt{4n}}$. Hence, $(\mathbb{I}\pm U)$ would be diagonally dominant matrix if $t < \sqrt{4n}$ and thus would be non-singular. That is why $(\mathbb{I} \pm U)$ would be invertible.
 
When all the entries of $U$ are identical i.e., $\frac{+1}{\sqrt{4n}}$ (or $\frac{-1}{\sqrt{4n}}$), then $(U^r)_{ij}=  \frac{t^{r-1}}{(\sqrt{4n})^r}$ (or $\frac{(-t)^{r-1}}{(\sqrt{4n})^r})$ respectively, else  we will have $|(U^r)_{ij}|\leq \frac{t^{r-1}}{(\sqrt{4n})^r}$. Now, for the convergence of the series for $(\mathbb{I} \pm U)^{-1} = \mathbb{I} \mp U + U^2 \mp U^3 + U^4 \ldots$, each entry of the series should converge. For this the sufficient condition is $\frac{t^{r-1}}{(\sqrt{4n})^r} < 1 \ \implies \ t < \sqrt{4n}$, which is also the condition for diagonal dominance. This diagonal dominance implies the invertibility of the matrix. Thus any square sub matrix $U$ of size $t < \sqrt{4n}$ of Hadamard matrix of order $4n$, will have valid series expansion for $(\mathbb{I} - U)^{-1}$ or $(\mathbb{I} + U)^{-1}$ in terms of the powers of $U$. Hence using Lemma \ref{ch6:Lem:Reduc} above we have $Y_1$ and $Y_2$ as $Y_1= D - W(\mathbb{I} + U)^{-1}V = \frac{1}{\sqrt{4n}}\left(\textsf{D} -\frac{1}{\sqrt{4n}} WV + \frac{1}{4n}WUV-\frac{1}{(4n)^{\frac{3}{2}}}WU^2V \ldots \right)$ and $Y_2 = D + W(\mathbb{I} - U)^{-1}V = \frac{1} {\sqrt{4n}}\left(\textsf{D} +\frac{1}{\sqrt{4n}} WV + \frac{1}{4n}WUV + \frac{1}{(4n)^{\frac{3}{2}}}WU^2V \ldots \right)$.

Now, in order to get the bound on $(Y_{1,2})_{ij}$, note that all the matrices $U, V, W, D$ have the entries $\pm 1$ only. Therefore, $|(WV)_{ij}| \leq t$, $|(WUV)_{ij}| \leq t^2$ and in general $|(WU^rV)_{ij}| \leq t^{r+1}$. Hence, the absolute value of the entries in the matrix would be bounded above, which implies $|(Y_{1,2})_{ij} | \leq \frac{1}{\sqrt{4n}}\left( 1+\frac{t}{\sqrt{4n}}+(\frac{t}{\sqrt{4n}})^2+(\frac{t}{\sqrt{4n}})^3\ldots \right) = \frac{1}{\sqrt{4n}} \left(1+ \frac{t}{\sqrt{4n} -t} \right)$ for $t < \sqrt{4n}$. Similarly the absolute values of the entries in the matrix $Y_{q}$ is bounded below, which would be given by  $|(Y_{q})_{ij} | \geq \frac{1}{\sqrt{4n}}\left( 1-\frac{t}{\sqrt{4n}}-(\frac{t}{\sqrt{4n}})^2-(\frac{t}{\sqrt{4n}})^3\ldots \right) = \frac{1}{\sqrt{4n}} \left(1- \frac{t}{\sqrt{4n} -t} \right)$. 
Now using $k = 4n-t$, we have $|(Y_{1,2})_{ij} | \leq \frac{1}{\sqrt{k+t}} \left(1+ \frac{t}{\sqrt{k+t} -t} \right) = \frac{1}{\sqrt{k+t} -t} = \frac{1+\epsilon}{\sqrt{k} }  $  where $\epsilon = \frac{\sqrt{k}}{\sqrt{k+t}-t} -1$. Hence in order to have $\epsilon \leq 1$ we get inequality $9k^2 +24kt + 16t^2(t-1)^2 - 40 kt^2 \geq 0$. Since $k,t > 0$, thus in this inequality only last term is negative. Hence in order to get a simple form, we may ignore the intermediate terms , then we get $9 k^2 - 40 kt^2 \geq 0 \Rightarrow t \leq \sqrt{\frac{9 k}{40}}< \sqrt{\frac{ k}{4}}< \sqrt{n}$. Thus when $t< \sqrt{n}$ the $Y_{1,2}$ are $\epsilon$-Hadamard matrices, with $\epsilon < 1$.  And in terms of the order ($k$) of the $\epsilon$-Hadamard matrices , we have    $|(Y_{1,2})_{ij} | \leq \frac{1}{\sqrt{k}}\left(1 + \frac{t}{\sqrt{k}} +\mathcal{O}\left( k^{-1} \right) \right)$. 
\flushright \qed
\end{proof}

As there are various options for real Hadamard matrices, when they exist, for $t \geq 1$, there are various possibilities of $U$ too and corresponding to each possibility, we will have different $Y_1, Y_2$. Thus, we should consider which one can be chosen for the best result among the  various possibilities of $U$,  we are considering. 
 
\subsection{Explaining the cases with $t = 1, 2, 3$}
\label{ch6:subsection:cases_for_hadamard}
We present explicit constructions of orthogonal matrices $Y_{(4n - t) \times (4n - t)}$ (i.e., $Y_1$ or $Y_2$ as in Theorem~\ref{ch6:th:ReducH}) by considering the block matrix $U_{t \times t}$. The matrix $U$ denotes a $t \times t$ matrix whose entries are all $\pm1$. 
For $t = 1$, we have $U = \pm1$. Given one element, without loss of generality, fix $U_1 = 1$. For $t = 2$, as there are 4 elements, we have $2^4 = 16$ possible configurations of $U$. For $t = 3$, there are $2^9 = 512$ possible configurations.

We have examined several (not all) cases out of these and here we present the results for the $\epsilon$-Hadamard matrix to have the least $\epsilon$ among the ones we considered below. As we have pointed out, $H_{4n}$ be an orthogonal matrix of the form
\[
H_{4n} = \frac{1}{\sqrt{4n}} 
\begin{bmatrix}
U_{t \times t} & V_{t \times (4n - t)} \\
W_{(4n - t) \times t} & D_{(4n - t) \times (4n - t)}
\end{bmatrix},
\]
where $t < \sqrt{4n}$ and all submatrices are real. Then the following cases are to be discussed one by one.

\subsubsection{The case $t=1$}
If $t=1$, i.e., $U = [1]$, then we construct
\[
Y_1 = \frac{1}{\sqrt{4n}} \left( D + \frac{1}{\sqrt{4n} - 1} W V \right), \quad
Y_2 = \frac{1}{\sqrt{4n}} \left( D - \frac{1}{\sqrt{4n} + 1} W V \right).
\]
Both $Y_1$ and $Y_2$ are orthonormal matrices of order $(4n - 1)$ with entries satisfying:
\[
|(Y_1)_{ij}| \in \frac{1}{\sqrt{4n}} \left\{ 1 + \frac{1}{\sqrt{4n} - 1}, \, 1 - \frac{1}{\sqrt{4n} - 1} \right\}, \quad
|(Y_2)_{ij}| \in \frac{1}{\sqrt{4n}} \left\{ 1 + \frac{1}{\sqrt{4n} + 1}, \, 1 - \frac{1}{\sqrt{4n} + 1} \right\}.
\]
$Y_2$ is the $\epsilon-$Hadamard matrix closest to Hadamard matrix, as $\epsilon$ corresponding to $Y_2$ is lesser than that of $Y_1$.

Here our construction yields an $\epsilon$-Hadamard matrix with $$\epsilon = \frac{\sqrt{4n-1}}{\sqrt{4n}}  \left(\frac{\sqrt{4n}+2}{\sqrt{4n}+1} \right)-1 < \frac{1}{2\sqrt{n}},$$ as the best one assuming $D, W$ and $V$ are any possible sub-matrices of a Hadamard matrix.

\subsubsection{The case $t=2$}
If $t = 2$, there are a total of $2^4 = 16$ possible configurations for the matrix $U$. When
\[
U \in \left\{
\begin{bmatrix}
1 & -1 \\
1 & 1
\end{bmatrix}, \quad
\begin{bmatrix}
1 & 1 \\
-1 & 1
\end{bmatrix}
\right\},
\]
we have $U^2 = 2U - 2\mathbb{I}$. This implies the following matrix inverses:

\begin{align*}
\left( \mathbb{I} + \frac{1}{\sqrt{4n}} U \right)^{-1} 
&= \frac{4n + 2\sqrt{4n}}{4n + 2\sqrt{4n} + 2}
   \left( \mathbb{I} - \frac{1}{\sqrt{4n} + 2} U \right), \\
\left( \mathbb{I} - \frac{1}{\sqrt{4n}} U \right)^{-1} 
&= \frac{4n - 2\sqrt{4n}}{4n - 2\sqrt{4n} + 2}
   \left( \mathbb{I} + \frac{1}{\sqrt{4n} - 2} U \right).
\end{align*}
Hence, we define the matrices:
\[
Y_1 = \frac{1}{\sqrt{4n}} \left( D - \frac{\sqrt{4n}+2}{4n+2\sqrt{4n}+2} W V + \frac{1}{4n + 2\sqrt{4n}+2} W U V \right),
\]
\[
Y_2 = \frac{1}{\sqrt{4n}} \left( D + \frac{\sqrt{4n}-2}{4n-2\sqrt{4n}+2} W V + \frac{1}{4n - 2\sqrt{4n}+2} W U V \right).
\]
Both $Y_1$ and $Y_2$ are orthonormal matrices of order $(4n-2)$. Here,$Y_2$ is the $\epsilon-$Hadamard matrix closest to Hadamard matrix, $\epsilon$ of $Y_2$ is lesser than that of $Y_1$. When
\[
U \in \left\{
\begin{bmatrix}
-1 & -1 \\
1 & -1
\end{bmatrix}, \quad
\begin{bmatrix}
-1 & 1 \\
-1 & -1
\end{bmatrix}
\right\},
\]
we have $U^2 = -2U - 2\mathbb{I}$. This implies the following matrix inverses:
\begin{align*}
\left( \mathbb{I} + \frac{1}{\sqrt{4n}} U \right)^{-1} 
&= \frac{4n - 2\sqrt{4n}}{4n - 2\sqrt{4n} + 2}
   \left( \mathbb{I} - \frac{1}{\sqrt{4n} - 2} U \right), \\
\left( \mathbb{I} - \frac{1}{\sqrt{4n}} U \right)^{-1} 
&= \frac{4n + 2\sqrt{4n}}{4n + 2\sqrt{4n} + 2}
   \left( \mathbb{I} + \frac{1}{\sqrt{4n} + 2} U \right).
\end{align*}
Hence, we define the matrices:
\[
Y_1 = \frac{1}{\sqrt{4n}} \left( D - \frac{\sqrt{4n}-2}{4n-2\sqrt{4n}+2} W V + \frac{1}{4n - 2\sqrt{4n}+2} W U V \right),
\]
\[
Y_2 = \frac{1}{\sqrt{4n}} \left(D + \frac{\sqrt{4n}+2}{4n+2\sqrt{4n}+2} W V + \frac{1}{4n + 2\sqrt{4n}+2} W U V \right).
\]
Both $Y_1$ and $Y_2$ are orthonormal matrices of order $(4n-2)$. $Y_1$ is the $\epsilon-$Hadamard matrix closest to Hadamard matrix, $\epsilon$ of $Y_1$ is lesser than that of $Y_2$.

One may note that, $Y_2$ of the first case, i.e., for $U^2 = 2U - 2\mathbb{I}$ and $Y_1$ in the second case, i.e., for $U^2 = -2U - 2\mathbb{I}$, have the same $\epsilon$ hence are two matrices are equally close to a Hadamard matrix.

In this case, our construction yields an $\epsilon$-Hadamard matrix with $$\epsilon = \frac{\sqrt{4n-2}}{\sqrt{4n}}  \left(\frac{4n+2}{4n-2\sqrt{4n}+2} \right)-1 < \frac{2}{\sqrt{n}},$$  as the best result assuming $D, W$ and $V$ are any possible sub-matrices of a Hadamard matrix.

\subsubsection{The case $t=3$}
If $t = 3$, there are a total of $2^9 = 512$ possible configurations for the matrix $U$, where each entry is either $+1$ or $-1$. When

\begin{equation*}
U = \left\{ \begin{split}
\begin{bmatrix}
1 & 1 &1\\
1 & -1 & 1\\
1 & 1 & -1\\
\end{bmatrix},  
& \begin{bmatrix}
-1 & 1 &1\\
1 & 1 & 1\\
1 & 1 & -1\\
\end{bmatrix}, 
 \begin{bmatrix}
-1 & 1 &1\\
1 & -1 & 1\\
1 & 1 & 1\\
\end{bmatrix}, 
 \begin{bmatrix}
-1 & -1 &1\\
-1 & 1 & -1\\
1 & -1 & -1\\
\end{bmatrix}, 
 \begin{bmatrix}
1 & -1 &1\\
-1 & -1 & -1\\
1 & -1 & -1\\
\end{bmatrix}, 
 \begin{bmatrix}
-1 & -1 &1\\
-1 & -1 & -1\\
1 & -1 & 1\\
\end{bmatrix}, \\
 \begin{bmatrix}
1 & 1 &-1\\
1 & -1 & -1\\
-1 & -1 & -1\\
\end{bmatrix},
& \begin{bmatrix}
-1 & 1 &-1\\
1 & -1 & -1\\
-1 & -1 & 1\\
\end{bmatrix},
 \begin{bmatrix}
-1 & 1 &-1\\
1 & 1 & -1\\
-1 & -1 & -1\\
\end{bmatrix}, 
\begin{bmatrix}
-1 & -1 &-1\\
-1 & -1 & 1\\
-1 & 1 & 1\\
\end{bmatrix}, 
 \begin{bmatrix}
-1 & -1 &-1\\
-1 & 1 & 1\\
-1 & 1 & -1\\
\end{bmatrix},
 \begin{bmatrix}
1 & -1 &-1\\
-1 & -1 & 1\\
-1 & 1 & -1\\
\end{bmatrix}, 
\end{split} \right\}
\end{equation*}

we have $U^3 = -U^2 + 4U + 4\mathbb{I}$. This implies the following matrix inverses:
\begin{align*}
\left( \mathbb{I} + \frac{1}{\sqrt{4n}} U \right)^{-1} 
&= \frac{4n\sqrt{4n} - 4n - 4\sqrt{4n}}{4n\sqrt{4n} - 4n - 4\sqrt{4n} + 4}
   \left( \mathbb{I} - \frac{\sqrt{4n} - 1}{4n - \sqrt{4n} - 4} U 
   + \frac{1}{4n - \sqrt{4n} - 4} U^2 \right), \\
\left( \mathbb{I} - \frac{1}{\sqrt{4n}} U \right)^{-1} 
&= \frac{4n\sqrt{4n} + 4n - 4\sqrt{4n}}{4n\sqrt{4n} + 4n - 4\sqrt{4n} - 4}
   \left( \mathbb{I} + \frac{\sqrt{4n} + 1}{4n + \sqrt{4n} - 4} U 
   + \frac{1}{4n + \sqrt{4n} - 4} U^2 \right).
\end{align*}

\noindent
Hence, we define the matrices:
\begin{align*}
Y_1 &= \frac{1}{\sqrt{4n}} \left(
        D 
      - \frac{4n - \sqrt{4n} - 4}{4n\sqrt{4n} - 4n - 4\sqrt{4n} + 4} W V 
      + \frac{\sqrt{4n} - 1}{4n\sqrt{4n} - 4n - 4\sqrt{4n} + 4} W U V \right.  \\
    &\quad \left.
      - \frac{1}{4n\sqrt{4n} - 4n - 4\sqrt{4n} + 4} W U^2 V 
      \right), \\
Y_2 &= \frac{1}{\sqrt{4n}} \left(
        D 
      + \frac{4n + \sqrt{4n} - 4}{4n\sqrt{4n} + 4n - 4\sqrt{4n} - 4}W V 
      + \frac{\sqrt{4n} + 1}{4n\sqrt{4n} + 4n - 4\sqrt{4n} - 4} W U V \right. \\
    &\quad \left.
      + \frac{1}{4n\sqrt{4n} + 4n - 4\sqrt{4n} - 4} W U^2 V 
      \right).
\end{align*}

Then $Y_1$ and $Y_2$ are orthonormal matrices of order $(4n - 3)$. $Y_1$ is the $\epsilon-$Hadamard matrix closest to Hadamard matrix, $\epsilon$ of $Y_1$ is lesser than that of $Y_2$.\\
When
\[
U \in
\left\{
\begin{aligned}
  &\begin{bmatrix}
    1 & -1 & 1\\
   -1 &  1 & 1\\
    1 &  1 & -1
  \end{bmatrix},\quad
  \begin{bmatrix}
    1 & 1 & 1\\
    1 & 1 & -1\\
    1 & -1 & -1
  \end{bmatrix},\quad
  \begin{bmatrix}
    1 & 1 & -1\\
    1 & 1 & 1\\
   -1 & 1 & -1
  \end{bmatrix},\quad
  \begin{bmatrix}
    1 & -1 & -1\\
   -1 &  1 & -1\\
   -1 & -1 & -1
  \end{bmatrix}, \\[1ex]
  &\begin{bmatrix}
    1 & -1 & 1\\
   -1 & -1 & 1\\
    1 &  1 & 1
  \end{bmatrix},\quad
  \begin{bmatrix}
    1 & 1 & 1\\
    1 & -1 & -1\\
    1 & -1 & 1
  \end{bmatrix},\quad
  \begin{bmatrix}
    1 & 1 & -1\\
    1 & -1 & 1\\
   -1 & 1 & 1
  \end{bmatrix},\quad
  \begin{bmatrix}
    1 & -1 & -1\\
   -1 & -1 & -1\\
   -1 & -1 & 1
  \end{bmatrix}
\end{aligned}
\right\}.
\]
we have  $\textsf{U}^3 = \textsf{U}^2+ 4\textsf{U} -4\mathbb{I}$. This implies the following matrix inverses:
\begin{align*}
\left( \mathbb{I} + \frac{1}{\sqrt{4n}} U \right)^{-1} 
&= \frac{4n\sqrt{4n} + 4n - 4\sqrt{4n}}{4n\sqrt{4n} + 4n - 4\sqrt{4n} - 4}
   \left( \mathbb{I} - \frac{\sqrt{4n} + 1}{4n + \sqrt{4n} - 4} U 
   + \frac{1}{4n + \sqrt{4n} - 4} U^2 \right), \\
\left( \mathbb{I} - \frac{1}{\sqrt{4n}} U \right)^{-1} 
&= \frac{4n\sqrt{4n} - 4n - 4\sqrt{4n}}{4n\sqrt{4n} - 4n - 4\sqrt{4n} + 4}
   \left( \mathbb{I} + \frac{\sqrt{4n} - 1}{4n - \sqrt{4n} - 4} U 
   + \frac{1}{4n - \sqrt{4n} - 4} U^2 \right).
\end{align*}
\noindent
Hence, we define the matrices:
\begin{align*}
Y_1 &= \frac{1}{\sqrt{4n}} \left(
        D 
      - \frac{4n + \sqrt{4n} - 4}{4n\sqrt{4n} + 4n - 4\sqrt{4n} - 4} W V 
      + \frac{\sqrt{4n} + 1}{4n\sqrt{4n} + 4n - 4\sqrt{4n} - 4} W U V \right.  \\
    &\quad \left.
      - \frac{1}{4n\sqrt{4n} + 4n - 4\sqrt{4n} - 4} W U^2 V 
      \right), \\
Y_2 &= \frac{1}{\sqrt{4n}} \left(
        D 
      + \frac{4n - \sqrt{4n} - 4}{4n\sqrt{4n} - 4n - 4\sqrt{4n} + 4}W V 
      + \frac{\sqrt{4n} - 1}{4n\sqrt{4n} - 4n - 4\sqrt{4n} + 4} W U V \right. \\
    &\quad \left.
      + \frac{1}{4n\sqrt{4n} - 4n - 4\sqrt{4n} + 4} W U^2 V 
      \right).
\end{align*}
Then $Y_1$ and $Y_2$ are orthonormal matrices of order $(4n - 3)$. $Y_2$ is the $\epsilon-$Hadamard matrix closest to Hadamard matrix, $\epsilon$ of $Y_2$ is lesser than that of $Y_1$.

One may note that, $Y_1$ of the first case and $Y_2$ of of the second case have the same $\epsilon$. Hence the two matrices are equally close to a Hadamard matrix.
In this case, our construction yields a $\epsilon$-Hadamard matrix with
\[
\epsilon = \frac{\sqrt{4n-3}}{\sqrt{4n}}
\left(
\frac{4n\sqrt{4n}+8n+2\sqrt{4n}+10}{4n\sqrt{4n}-4n-4\sqrt{4n}+4}
\right)
- 1 < \frac{4}{\sqrt{n}}, for\ n\ \geq 4,
\]
as the best result assuming $D$, $W$ and $V$ are any possible submatrix of a Hadamard matrix.

We can summarize these in the following results.
\begin{theorem}
\label{ch6:epsathm}
For $t = 1, 2, 3$, it is possible to construct real $\epsilon$-Hadamard matrices of dimension $4n-t$ with $\epsilon = \frac{\rho_t}{\sqrt{n}}$, where $\rho_t = \frac{1}{2}, 2$ and $4$ respectively for $t = 1, 2, 3$.
\end{theorem}
Programs in Python are implemented in this regard to experiment with such $\epsilon$-Hadamard matrices, and the repository is available at~\cite{github}. We have noted that the actual values are even smaller than the expressions in Theorem~\ref{ch6:epsathm}, but the expressions will be more complicated.

\subsection{Construction of ARMUBs}
As we have already discussed, construction for Approximate MUBs using Hadamrd and other unitary matrices in conjunction with Resolvable Block Designs have been presented in~\cite{ak21,ak22,ak23}. The basic concept is to utilize an RBD$(X, A)$, with $|X|= d = k \times s$, having constant block size $k$, such that it has as many parallel classes as possible, but any pair of block from different parallel classes, should have maximum $1$ point in common $(\mu=1)$ (as we have discussed in Section~\ref{ch6:rbdsec}). Here we need to use the unitary matrix of the order of block size $k$, to convert each parallel class into an orthonormal basis.

Given this context, we present the main result as follows.
\begin{theorem}
\label{ch6:ARMUB}
Let $d = k \times s$, such that $s > k$ and $s$ is a power of odd prime. Let $\epsilon$-Hadamard matrix $Y$ of order $k$ exist. Then one can obtain $s \geq \lceil \sqrt{d}\rceil$ many $\beta$-ARMUBs, with 
$\beta \leq 1 + O(d^{-\frac{1}{4}})$, which is $< 2$ for an infinite class of dimensions $d$.
\end{theorem}
\begin{proof}
We broadly consider Construction~\ref{ch6:cons:1}, to obtain the orthonormal basis vectors in $\mathbb{R}^d$ corresponding to each parallel class in RBD$(X, A)$ using $\epsilon$-Hadamard matrix $Y$. Following~\cite{ak22,ak23} the absolute value of the dot product of any pair of basis vectors $u, v$ from different orthonormal basis $|\braket{u, v}| \leq \mu |(Y)_{ij}|_{\max}^2 \leq (\frac{1+\epsilon}{\sqrt{k}})^2 = \frac{(1+\epsilon)^2}{k}$, where $\mu$ is the intersection number and in this case any pair of blocks from different parallel classes will have maximum one point in common, i.e., $\mu = 1$. When $s-k > 0$ and $s$ a power of prime, then one can construct RBD$(X,A)$ having $s$ many parallel classes, each with $s$ many blocks of constant block size $k$ (refer to~\cite[Construction 3, Lemma 5]{ak23} for exact details). 

With proper manipulations of constants and noting the power of primes are at least as dense as the primes, we will obtain an odd prime power $s > \sqrt{d}$, in an expected distance $O(\log d)$ from $\sqrt{d}$. Thus, considering Theorem~\ref{ch6:epsathm} for the values of $\epsilon$, we obtain the inner product value between two vectors from two different ARMUBs  
$\leq \frac{(1+\epsilon)^2}{k} \leq \frac{1}{\sqrt{d}}(1 + O(d^{-\frac{1}{4}}))$. This gives $\beta \leq 1 + O(d^{-\frac{1}{4}})$. \qed
\end{proof}

\begin{example}
Let us now consider an example for $d = 79 \times 3^4$, where $k=79, s = 3^4$. In this case there are $79+1 = 80$ complex MUBs, but no pair of real ones. The existing works in this regard~\cite{ak21,ak22,ak23} do not consider this case too. In this effort, we will obtain $3^4 = 81$ many ARMUBs, which is $\lceil \sqrt{d} \rceil + 1$. As $79 = 4 \times 20 - 1$, from Theorem~\ref{ch6:epsathm}, we have $\epsilon < \frac{1}{2\sqrt{n}} = \frac{1}{2\sqrt{20}}$. Now, $k = 4n-1 = 79$, thus, the numeric expression will be $\beta < 1.252$.

In case, we consider $k=43$, instead of 79, we obtain $\beta \leq 1.697$. Naturally, we obtain the value of $\beta$ closest to 1 (from above), when $k, s$ are very close. 
\end{example}

\section{Conclusion}
In this work, we presented constructions of $\beta$-ARMUBs for certain dimensions $d$ that are not multiples of $4$, when $d$ is of the form $(4n-t)s$, where $s$ is an odd prime power. Specifically, we obtain $\beta = 1 + O(d^{-\frac{1}{4}}) < 2$, which is the upper bound of the inner product between any two vectors from two different MUBs. The number of ARMUBs available for these cases are $\geq \lceil \sqrt{d} \rceil$. Such classes of ARMUBs, for dimensions not divisible by 4, had not be presented earier.
However, our constructions do not yield real APMUBs (Almost Perfect Real MUBs or APRMUBs) since the set $\Delta$ in our case cannot guarantee $|\Delta| = 2$, which is mandatory for an APMUB construction, i.e., in our case, there may be more than 2 values of inner products. As no pair of real MUBs exists for dimensions $d \not\equiv 0 \bmod 4$, such constructions for APRMUBs would be an interesting research direction.  

\bibliographystyle{plain}

\end{document}